%% file: mobiopp.tex
\documentclass[letter]{acm_proc_article-sp}

\usepackage{url}
\usepackage{graphicx,theorem}
\usepackage{amssymb,subfigure}
\usepackage{color}
\usepackage{multirow}

\newtheorem{theorem}{Theorem}[section]

\usepackage{tabularx}

\newcolumntype{M}[1]{>{\centering\arraybackslash} p{#1}}

\begin{document}

\title{Modelling the Delay Distribution of Binary Spray and Wait Routing Protocol}

\numberofauthors{2} 
\author{
\alignauthor 
R\'emi DIANA\\
       \affaddr{Univ. de Toulouse, CNES, Thales Alenia Space, France}\\
       \email{remi.diana@isae.fr}
\alignauthor 
Emmanuel LOCHIN\\
       \affaddr{Univ. de Toulouse, ISAE, TeSA, Toulouse, France}\\
       \email{emmanuel.lochin@isae.fr}
}
\date{18 November 2011}

\maketitle

\begin{abstract}
This article proposes a stochastic model to obtain the end-to-end delay law between two nodes of a Delay Tolerant Network (DTN). We focus on the commonly used Binary Spray and Wait (BSW) routing protocol and propose a model that can be applied to homogeneous or heterogeneous networks (i.e. when the inter-contact law parameter takes one or several values). To the best of our knowledge, this is the first model allowing to estimate the delay distribution of Binary Spray and Wait DTN protocol in heterogeneous networks. We first detail the model and propose a set of simulations to validate the theoretical results.
\end{abstract}

\begin{keywords}
DTN routing, modelling, binary spray and wait
\end{keywords}

\section{Introduction}

Delay Tolerant Networks (DTN) is a concept initially created for interplanetary networks \cite{IPN}. However, it also receives a great success for intermittently connected networks and particularly 
for opportunistic networks \cite{ccr08}.
In these networks, a node can send data to another if both are in the same transmission range. Due to the dynamic character of these networks, there is no guarantee that a direct connected path from a given source to a given 
destination exists at any time. As a result, routing protocols using relay nodes and replication such as MaxProp \cite{maxprop}, Spray and Wait \cite{sw}, PRoPHET \cite{prophet} and RAPID \cite{rapid} have been proposed to increase the message delivery ratio over such intermittently connected networks.

The performance evaluation of relay protocols in terms of message delivery ratio, end-to-end delivery delay or throughput is a difficult task due to the complexity to drive mobile network simulations. Several efforts have been done in order to assess the performance of routing schemes with simulations. Today, The ONE simulator became a referent tool in this area \cite{theone}. Other approaches have proposed Markovian and ordinary differential equations (ODEs) models to study the performance of some 
basic routing protocols such as Epidemic, Epidemic limited, 2-hop routing and 2-hop limited routing protocols \cite{model_epidemic_ode}, \cite{model_epidemic_2hop} while others focus on the ressource constraints 
issues in these networks \cite{model_energy}, \cite{model_energy2}. However, all these models do not consider both Binary Spray and Wait (BSW) routing protocol and different inter-contact law parameters (called in this study heterogeneous case).

In this paper, we introduce a Markovian model to obtain the end-to-end delivery delay law and the average delivery ratio of an intermittently connected network. Compare to previous existing works, we propose to fill a gap by introducing a 
model of the commonly-used Binary Spray and Wait routing protocol in both homogeneous and heterogeneous cases. Indeed, in most DTN routing studies, this protocol is used as a reference for comparison purpose as BSW has been proved 
to be optimal in a fully random network \cite{sw}. To the best of our knowledge this is the first model proposed for BSW performances. Section \ref{sec:assumptions} presents and justifies the assumptions chosen and sums up
the notations used inside this paper. In Section \ref{sec:BSW_model_homo}, we first propose a BSW model for the homogeneous case. This model is then extended to handle heterogeneous networks in Section \ref{sec:BSW_model_hetero}.
In each section we provide examples to assess the consistency and efficiency of the developed model and compare the results obtained with The ONE simulator. Section \ref{sec:conc} concludes this work and details the future work.


\section{Assumptions and notations}
\input{assumptions}

\section{Model of Binary Spray and Wait routing protocol for homogeneous network}
\label{sec:BSW_model_homo}
\input{model_homo}

\section{Extension of BSW model to heterogeneous networks}
\label{sec:BSW_model_hetero}
\input{model_hetero.tex}

\section{Conclusion}
\label{sec:conc}
In this article, we propose a model to assess the end-to-end delay in an intermittently connected network using Binary Spray and Wait routing protocol. Under the assumption of an exponential inter-contact time distribution, we give a Markov chain that represents the diffusion of message copies in the network. This representation allows to obtain the end-to-end delay $D$, as the solution function to the first hitting time theorem. The extended version of this model allows to deal with the case of heterogeneous networks. As explained in Section \ref{sec:BSW_model_hetero}, we give the rules to build a Markov chain using a contact matrix of the network. The end-to-end delay $D$ remains the solution of the first hitting time theorem. We drive a set of simulations that confirm the accuracy of the model. We also verified the accuracy of our model on more realistic cases, both a $12$-node and $20$-node heterogeneous networks. 

We do not present estimation in terms of computational efficiency of the model. However, the first implementation realised in this paper (available on the author's webpage) demonstrates a faster resolution twice or three times
faster than the corresponding simulation with The One for simulation involving 20 nodes in the heterogeneous cases. In a future work, we expect to drive sereval experiments to assess the exact cost in terms of computation.
We currently investigate the integration of VACCINE \cite{vaccine} inside the model in order to determine an average amount of buffer occupancy and an achievable throughput.




\section{Acknowledgements}

The authors would like to thank CNES and Thales Alenia Space for their support.

\bibliographystyle{plain}
\bibliography{biblio}

\end{document}

%% file: assumptions.tex
Before presenting the assumptions used to build our model, we first recall how the BSW routing protocol operates. 

The source node of a message initially starts with a fixed number of copies denoted $L$. This number is called the replication factor. Then, the spray phase 
is directed by the following rule: any node that has strictly more than one message copy (source or relay) gives half of its copies 
to the first node (without copies) encountered. When a node has only one copy, it switches to the wait phase and give its copy to and only to the destination.

\subsection{Assumptions}
\label{sec:assumptions}

Our model is based on two main assumptions:
\begin{enumerate}
\item the model does not consider buffer constraints (i.e. losses resulting from congestion) and losses due to link failure. That means that we model a case where each contact is long enough to send and/or receive all required messages. Note that the case of congestion is discussed later in Section \ref{sec:conc}; 
\item we consider all inter-contact laws as exponential. Following \cite{dual_law}, the authors show that the time scale of interest for opportunistic forwarding may be of the same order as the characteristic time, and thus the exponential tail is important. As a result, the exponential distribution of inter-contact is meaningful and justifies a Markovian model. In this paper, the authors also claim that the choice of a power law (as proposed in \cite{powerlaw}) in these cases leads to pessimistic results. The use of exponential laws is clearly justified, however it would be interesting to qualify and quantify the error done with such an assumption in a case of network characterized by different inter-contact laws. This problem will be tackled in a future work.
\end{enumerate}

\subsection{Notations}
\label{sec:notations}

We consider a network with $N$ nodes, noted $n_i, \ i \in \lbrace 1,..,N \rbrace$. $ \forall (i,j) \in \lbrace 1,..,N \rbrace \times \lbrace 1,..,N \rbrace, \ i\neq j$, the inter-contact law between $n_i$ and $n_j$ is an exponential law of parameter $\lambda_{i,j} = \lambda_{j,i}$. In our study, we also consider homogeneous networks that means $ \forall (i,j) \in \lbrace 1,..,N \rbrace \times \lbrace 1,..,N \rbrace, \ i\neq j, \ \lambda_{i,j}=\lambda$. Thus, there is only one parameter: $\lambda$. Previous notations are summed up in Table \ref{tab_notation}.

\begin{table}[h]
	\caption{Notations used for homogeneous and heterogeneous models.}
	\begin{center}
	\begin{tabular}{|c|c|} 
		\hline Notation & Definition \\
		\hline N & amount of nodes in the network \\
		\hline $i$ & index of nodes \\
		\hline $n_i$ & $i^{th}$ node of the network \\
		\hline \multirow{2}{*}{$\lambda_{i,j}$} & parameter of the exponential inter-contact \\
											   &  law between nodes $n_i$ and $n_j$\\
		\hline \multirow{2}{*}{$\lambda$} & for homogeneous networks all inter-contact \\
										 & laws have the same parameter  \\
		\hline $L=2^k$ & replication factor of BSW routing protocol \\
		\hline
	\end{tabular}
	\label{tab_notation}
	\end{center}
\end{table}

%% file: model_homo.tex
The model is done in two parts. First, we build a Markov chain representing the dissemination of copies in the network with an absorbing state corresponding to the delivery of the message. Then, we apply the first hitting time theorem \cite{hitting} between the initial state representing the creation of the message by the source and the absorbing state. This theorem gives the distribution of time needed to reach the absorbing state starting from the first state. In other words, this corresponds to the end-to-end delay between a given source and destination. The main issue is to create a Markov chain that represents the BSW routing protocol. 

In the following, we consider that each node can be in contact with all other nodes with an identical inter-contact law parameter. We qualify this network as homogeneous.

\subsection{Markov Chain for homogeneous cases}

We define a state of the Markov chain as a possible repartition of messages in the network. For example, a possible repartition for a replication factor of $8$ can be: one node with $4$ copies, one node with $2$ copies and two nodes with $1$ copy. We consider that the number of replicates is a power of two, $2^k$. However, the methodology described in the rest of the paper is easily adaptable to any replication factor $L$.

\begin{theorem}
\emph{Number of states in the Markov Chain} \\
\label{nb_state_theo}

In a $N$-node homogeneous DTN, using Binary Spray and Wait routing protocol with a replication factor $L=2^k$, the number of states is:
$$N_{states} = \beta(k) + 1$$ 
with $\beta(k)$ the number of partitions of $2^k$ into powers of 2.

\end{theorem}

\begin{proof}
A state corresponds to a particular repartition of copies into the network. A forwarding node, according to BSW protocol, gives half of its copies until it finally gets only one. Thus, each node can have a number of copy in $\{1,2,..,2^k\}$. Moreover, we do not need to discriminate the nodes between them since we consider an homogeneous network. Thus, the number of different possible repartition is the number of partitions of $2^k$ into powers of 2 denoted $\beta(k)$. As we focus on the delay of the first copy reaching the destination, we add an absorbing state which represents the final delivery of the copy of a message. Thus, the number of states is $\beta(k)+1$.
\end{proof}

We provide in Table \ref{nb_states_array} the number of states for different values of $L$. We remark that these results are true for $L<N$. 

\begin{table}[h]
	\caption{Value of $\beta$ sequence and corresponding number of states as a function of $L$}
	\begin{center}
	\begin{tabular}{|M{2cm}|M{2cm}|M{2cm}|} 
		\hline $L$ & $\beta(k)$ &  $N_{states}$ \\
		\hline   2 &   2 &   3 \\
		\hline  4 &  4 &  5 \\
		\hline  8 &  10 &  11 \\
		\hline  16 &  36 &  37 \\
		\hline  32 &  202 &  203 \\
		\hline  64 &  1828 &  1829 \\
		\hline  128 &  27339 &  27339 \\
		\hline  256 &  692004 &  692005 \\
		\hline
	\end{tabular}
	\label{nb_states_array}
	\end{center}
\end{table}


\subsubsection{Transition in the chain}

We have computed the number of states in the Markov chain. We now have to detail how to compute the transition parameters. 

\begin{theorem}
\emph{Minimum of $n$ exponential laws} \\
\label{min_n_exp_theo}
Let $\{X_i\}_{i \in \{1,..,n\}}$ be $n$ random variables following exponential laws of respective parameter $\lambda_{X_i}$.\\
Let $Z=Min_{i\in\{1,..,n\}}X_i$. 
Then, $Z$ is a random variable following an exponential law of parameter $\lambda_Z= \sum_{i=1}^n \lambda_{X_i}$.
\end{theorem}

There is two type of transitions: 
\begin{itemize}
	\item transition from one state to the absorbing state;
	\item transition from one state to another one.
\end{itemize}

The expression of the transition parameter between one state and the absorbing state depends on the number of nodes that have a copy of the message. We denote this number: $n_p$. Each of these $n_p$ nodes can join the destination. The destination is reached as soon as one of these $n_p$ nodes is in contact with it. Thus, the law of the transition is given by the minimum of $n_p$ exponential laws of parameter $\lambda$ which is $n_p\lambda$. We can differentiate two cases: either the source can be in contact with the destination (WDC: with direct contact) or can not (NDC: no direct contact). Nevertheless, as the source always keeps at least one copy of the message, the transition parameter can be written as follows: \\
$ n_p\lambda $, WDC  or $ (n_p-1)\lambda $, NDC.

To compute a normal transition, we first have to focus on the partition. Indeed, copies repartition corresponds to: 
$$L=2^k=\sum_{j=0}^k a_j 2^j$$ 
where $a_j$ represents the number of nodes that have $2^j$ copies of the message. This partition can also be written as a vector: 
$$(a_j)_{j\in\{0,..,k\}}$$
We consider $(a_j)_{j\in\{0,..,k\}}$ and $(b_j)_{j\in\{0,..,k\}}$ two repartition of copies, respectively of states A and B. We suppose that we transit from A to B when a node with $2^{m}$ copies is in contact with a node with no copies. The relationship between A and B can be written as follows: 
$$b_{m}=a_{m}-1 \text{ and } b_{m-1}=a_{m}+2, m \in\{1,..,n\}$$ 
with $m \geqslant 1$ since a node with one copy can forward this last remaining copy only to the destination (i.e. the corresponding state is the absorbing state). 
Keeping the previous notations, we can express $n_p$ (the number of node that have a copy of the message) as follows: 
$$n_p=\sum_{j=0}^k a_j$$ 
The transition between states A and B is done because a node with $2^{m}$ copies gives to another node $2^{m-1}$ copies. This node can give these copies to $N - n_p - 1$ different nodes since we do not consider the destination (which is a particular state in the chain). In practice, it gives these copies to the first one met. Thus, the law of the transition corresponds to the minimum of $N-n_p - 1$ exponential laws of parameter $\lambda$. Moreover, to make the transition from A to B, only one node among $a_{m}$ nodes must give half of its copies. Thus, the law of the transition corresponds to the minimum of $a_{m}(N-n_p - 1)$ exponential laws of parameter $\lambda$. As a consequence the transition
parameter is $a_{m}(N-n_p - 1)\lambda$.

All transitions parameters have to be positive. If $L \geqslant N$, some states are unreachable, become senseless and should be removed.

The Markov chain is now built and complete since we have the number of states in the chain and the literal expression of all transitions. The second phase consists in applying the first hitting time theorem \cite{hitting} between the initial state (where the source has all the message copies) and the absorbing state (corresponding to the delivery of the message) in order to obtain the delay distribution law.

\subsection{Practical examples and simulations}
\label{sec:simu_homo}
In this section, we give a representation of some Markov chains. We present complete Markov chains for $L=4$ (Figure \ref{homo_L4}), $L=8$ (Figure \ref{homo_L8}) and $L=16$ (Figure \ref{homo_L16}). These three chains correspond to a NDC case. This means there is no transition between the first state and the absorbing state.

\begin{figure}[ht!] 
	\centering
	\subfigure[BSW with L=4, homogeneous case]{\includegraphics[width=0.8\columnwidth]{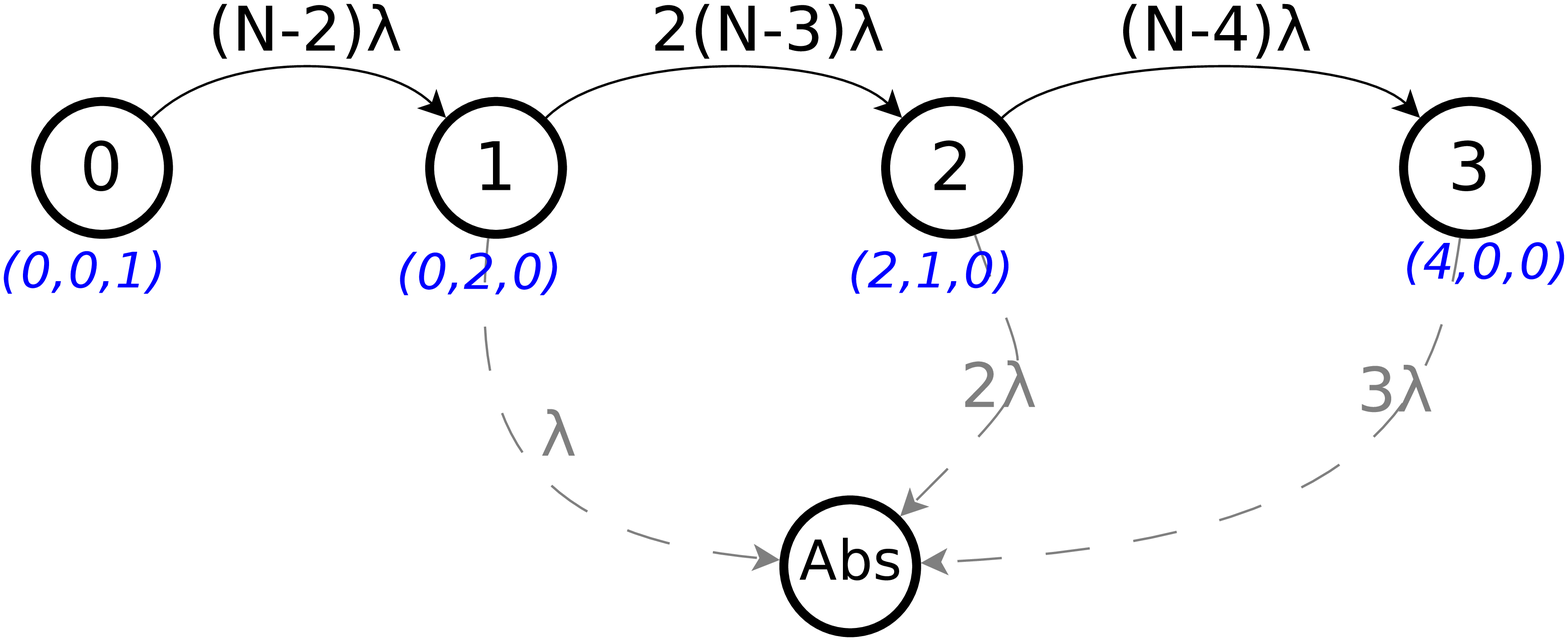}\label{homo_L4}}
	\subfigure[BSW with L=8, homogeneous case]{\includegraphics[width=0.98\columnwidth]{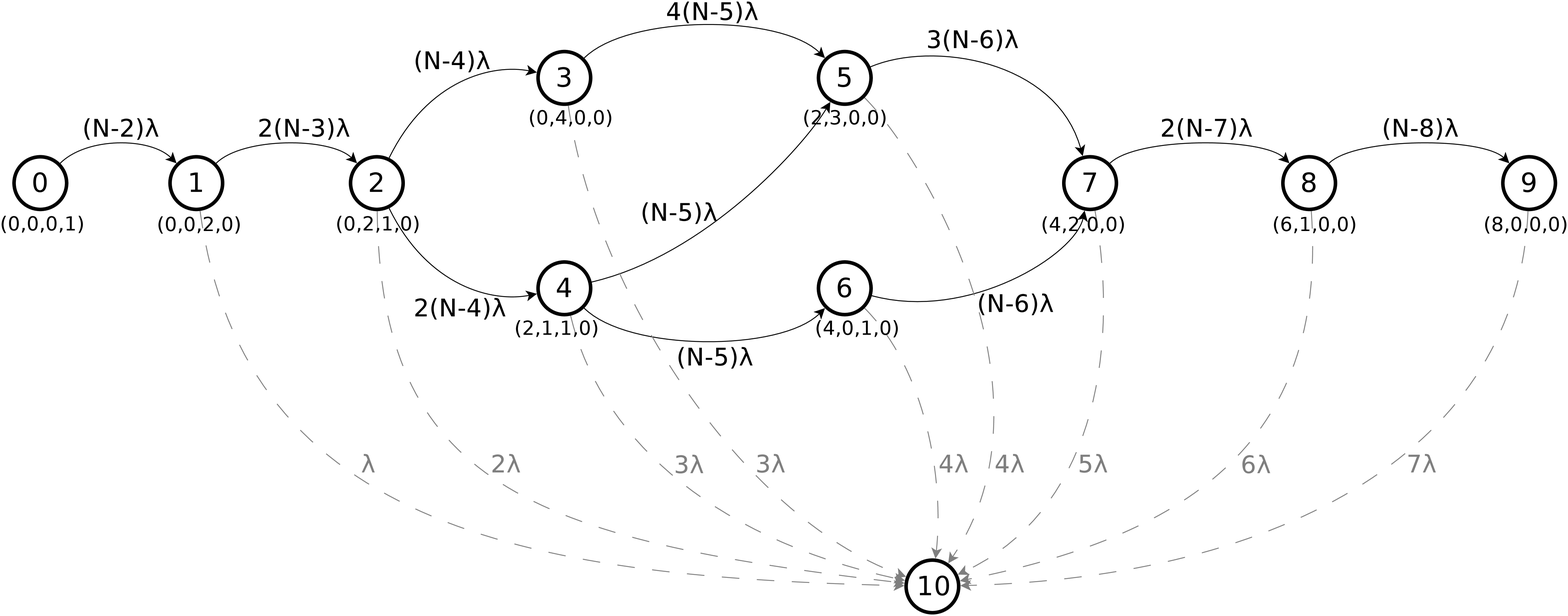}\label{homo_L8}}
	\caption{Example of Markov Chains for homogeneous network with $L=4$ and $L=8$ (the corresponding repartition is indicated inside the states)}
	\label{Markov_chains}
	
\end{figure} 

\begin{figure*}[ht!] 
	\centering
	\includegraphics[width=1.98\columnwidth]{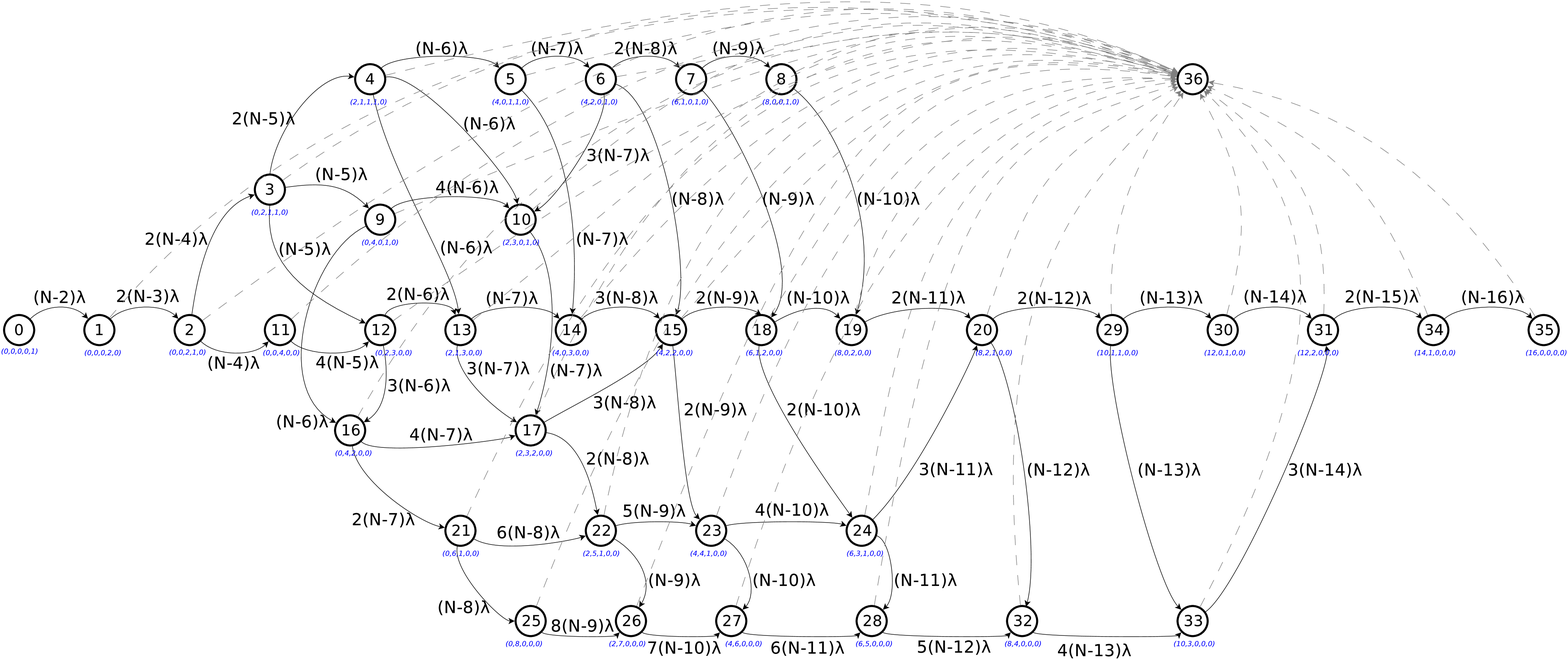}
	\caption{Example of Markov Chain for homogeneous network with $L=16$ (the corresponding repartition is indicated inside the states)}
	\label{homo_L16}
	
\end{figure*} 	

We use previous Markov chains definition to validate our model for different values of $L$ and $N$.
Table \ref{tab:bsw_exp} summarizes the different cases evaluated and gives the main network parameters.

\begin{table}[h]
	\caption{The different cases simulated and compared to the analytical results obtained with the model.}
	\begin{center}
	\begin{tabular}{|M{1.5cm}|M{1.5cm}|M{1.5cm}|M{1.5cm}|} 

		\hline  Case  &  $L$ &  $N$ & $\lambda$ \\
		\hline  \#1 &  4 &  6 &  50 \\
		\hline  \#2 &  4 &  20 &  200 \\
		\hline  \#3 &  8 &  20 &  200 \\
		\hline  \#4 &  16 &  20 &  200 \\
		\hline
		\end{tabular}
		
	\label{tab:bsw_exp}
	\end{center}
\end{table}

We use The ONE simulator \cite{theone} to perform our simulations. To evaluate $D$, the random variable corresponding to the end-to-end delay of messages, we first create a contact trace file of several millions of seconds with correct parameters of inter-contact laws. Using this file, the simulation consists in the sending of $N_e$ messages by the source. Once a message is created, the diffusion process starts. The messages generation is sufficiently spaced to ensure that each message transmitted from a source to a destination is an independent event. In practice, we choose the delay between two messages sending greater than $\lambda$. Thus, we observe $N_e$ independent events of the random variable $D$. In all our cases, this permits to accurately evaluate the distribution of $D$. However, it is easy to increment the accuracy of this evaluation by increasing the number of events observed. In our experiments, $N_e$ is ranging from $2000$ to $10000$.

Theoretical results have been obtained using Matlab. Figure \ref{homo_BSW_20n_result}, presents both simulation and theoretical results for the four cases described in Table \ref{tab:bsw_exp}. This figure gives the results for a 20-node network with a replication factor ranging from $4$ to $16$. We observe that the results obtained by our model correctly fits the corresponding simulation.

\begin{figure}[ht!] 
	\centering
	\includegraphics[width=0.98\columnwidth]{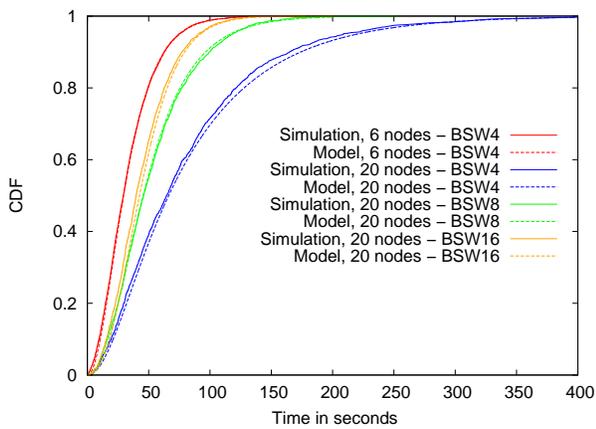}
	\caption{Comparison of the results obtained by simulation and our model in different cases for homogeneous network.}
	\label{homo_BSW_20n_result}
	
\end{figure} 

Following these results, we now propose to extend this model to the heterogeneous case.

%% file: model_hetero.tex
As explained in Section \ref{sec:BSW_model_homo}, in a homogeneous network, there is no discrimination between nodes with a given number of copies. Indeed, if two nodes have $4$ copies of a message, it does not matter to distinguish them as this is not taken into account for the computation of the transition in the Markov chain. On the contrary, in a heterogeneous network, we need to distinguish these nodes to compute the transitions and the partitions of the network. A partition where two nodes have $4$ copies of a message is not the same as each node can have different inter-contact law with all the other nodes. The problem is now to build a Markov chain that takes into account this new state. Obviously, the number of states in the chain is going to increase and depends on $N$. However, this model allows to assess the delay distribution of BSW routing protocol in any intermittently connected networks where nodes have contacts only with a subset of other nodes. 

\subsection{Markov Chain for heterogeneous cases}

This new Markov chain can be seen as a generalization of the previous one proposed for homogeneous case. Instead of a vector used to represent the copies repartition, we now use a matrix. Basically, each line of this matrix represents a node of the network and each column represents a number of copies in the same way as the vector in the previous part. We denote $R=(r_{i,j})_{1\leqslant i \leqslant N-1 ; 1\leqslant j \leqslant k+1}$ the copies repartition. $R$ has only $N-1$ lines as the destination is not considered in the repartition. If we consider a vector $V$ defined as follows: $V=(v_{i})_{1\leqslant i \leqslant k+1}$ with $v_i=\Sigma_{m=1}^{N-1}r_{m,i}$, this vector can be seen as a repartition of copies in a homogeneous network. As a result, the heterogeneous Markov chain corresponds to an extension of the homogeneous one which consists in splitting homogeneous states in several part to allow nodes discrimination. Transitions from one given state to the absorbing state are computed in the same way as in the homogeneous case while no computation is needed for the other transitions. 

The number of links denoted $n_l$, that starts from a given state in the heterogeneous case is equal to:
$$\sum_{j=1}^{k}(N-1-n_p)n_r(j)$$
with $n_p$ the number of relays that have a copy and $n_r(j)$ the number of relays that have $2^j$ copies. The exact number of states, which is not trivial to obtain, is computed with Matlab. 
To illustrate how the problem is finally solved, we give an example of the Markov chain obtained for the case $N=5$ and $L=2^2$ in Figure \ref{fig_hetero_N5_L4} with direct contact.
\begin{figure*}[ht!] 
		\begin{center}	
		\includegraphics[width=1.5\columnwidth]{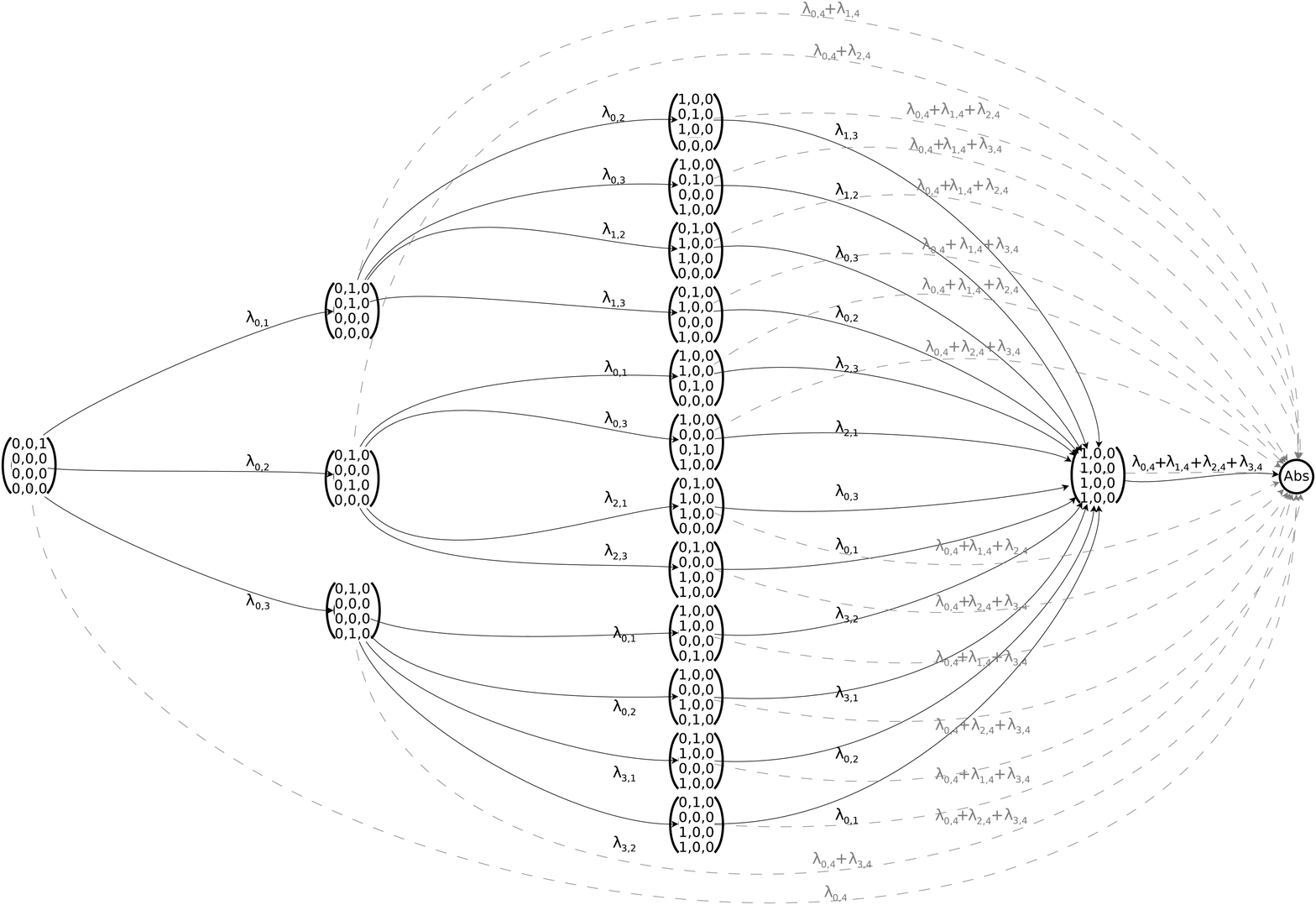} 
		\caption{A practical example of chain building in heterogeneous case.} 
		\label{fig_hetero_N5_L4} 
		\end{center}
\end{figure*} 
As a potential application example of the previous formula, we consider the second state of the chain: 
\tiny$ \begin{pmatrix}
	0&1&0 \\
	0&1&0 \\
	0&0&0 \\
	0&0&0 \\
\end{pmatrix}$ \normalsize.
There is only one kind of transition which is a transition from a node with two copies ($n_r(2)=0$). Here, $n_p=2$, $n_r(1)=2$ so $n_l=4$. This means that this particular state generates four different other states (as shown in Figure \ref{fig_hetero_N5_L4}).

In the case of $L=4$, we can give a literal expression of the number of states in the chain based on Figure \ref{fig_hetero_N5_L4}. Each level in the chain corresponds to the number of nodes that have a message. A state of the second level is a state where the source has two copies and one node among the $N-2$ remaining nodes. There is $\binom{N-2}{1}$ different possible states. A state of the third level is a state where a node has two copies and two nodes have one copy. Thus, there is $(N-1).\binom{N-2}{2}$ different possible states. A state of the last level is a state where four nodes have one copy, but in all states the source will have one copy. There is $\binom{N-2}{3}$ different states for this last level. Finally, for $L=4$, the number of states is given by: 
\begin{equation}
2 + \frac{N-2}{6}(6+(N-3)(4 N-7)) 
\label{eqL4}
\end{equation}
Note that if some nodes are never in contact, some transitions are not possible and some states are unreachable.
 	
For $N=5$ or $N=10$ the chain has respectively $18$ or $318$ states. We observe that this number fastly increases as a function of $N$. This trend will be even more significant when $L$ also increases. However, in a heterogeneous case which fairly represent a real case, many transitions will be null since some nodes will never meet some other ones. As a result, the matrix that represents the Markov chain has a large dimension but remains very sparse and can be computed. We have developed an algorithm to compute the states and the transitions between them.

\subsection{Practical example of heterogeneous cases}

In this section, we present three experiments with heterogeneous networks.

\subsubsection{Case \#1}
In this first example, we take a simple network composed by five nodes with $L=4$. We also set $\lambda_{1,2}=\lambda_1$, $\lambda_{1,3}=\lambda_2$, $\lambda_{1,4}=\lambda_3$.  All other parameters are equal to $\lambda$ and we suppose there is no direct contact.
The Markov chain is the same that the one presented in Figure \ref{fig_hetero_N5_L4} with these corresponding values of $\lambda_{i,j}$

We compare both simulation and theoretical results obtained with this model. For the experiment, we choose $\lambda_1=100$, $\lambda_2=200$, $=\lambda_3=500$ and $\lambda=200$. Simulation are driven as explained in Section \ref{sec:simu_homo} except that we choose a delay between two messages at least as long as the largest parameter of the inter-contact laws. Results are presented in Figure \ref{hetero_result}.

\begin{figure}[ht!] 
		\begin{center}	
		\includegraphics[width=0.85\columnwidth]{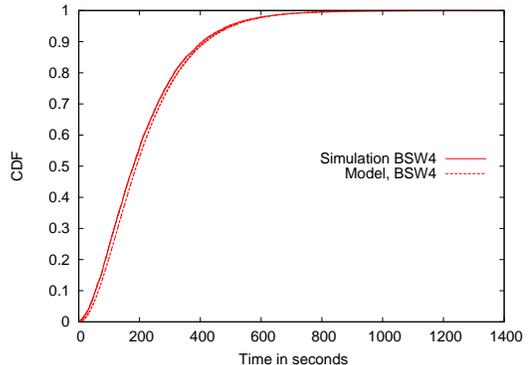} 
		\caption{Results obtained with case \#1.}
		\label{hetero_result} 
		\end{center}
\end{figure} 

\subsubsection{Case \#2}
In this second heterogeneous case, we approach a more realistic scenario. Indeed, we choose a set of $12$ nodes. Each node has an immediate number of neighbours (called \textit{diversity} in the following) ranging from 2 to 8. 
The parameter of each inter-contact law is randomly set between $200$ and $1200$ seconds. The chosen network is the sub-network made of the $12$ first nodes of the $20$-node network presented in Figure \ref{graph_20n}. 
We compare the theoretical and simulated end-to-end delay distributions for $L=4$ and $L=8$. Results are presented in Figure \ref{results_20n}.

\begin{figure}[ht!] 
		\begin{center}	
		\includegraphics[width=0.85\columnwidth]{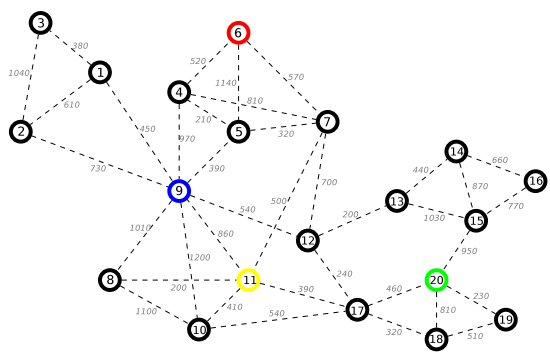} 
		\caption{$20$-node network representation.} 
		\label{graph_20n} 
		\end{center}
\end{figure}

\subsubsection{Case \#3}
This last heterogeneous case consider the whole $20$-node network presented in Figure \ref{graph_20n}. Each node still has a diversity ranging from $2$ to $8$ and inter-contact law parameter is also represented in Figure \ref{graph_20n}.
We compare theoretical and simulation results of end-to-end delay distribution for $L=4$ and $L=8$ in this case. 

Figure \ref{results_20n} presents the results for both cases (i.e. cases \#2 and \#3). Solid lines correspond to simulation results while dotted lines to the theoretical ones. 
We observe that the results obtained by our model fairly fit those obtained by simulation. Moreover, the third case illustrates that the model also captures the fact that 
the delivery ratio does not always reach 100\%. Indeed, the average delivery ratio in the third case is $18 \%$ for $L=4$ and $57 \%$ for $L=8$ which is accurately captured by the model.

\begin{figure}[ht!] 
		\begin{center}	
		\includegraphics[width=0.85\columnwidth]{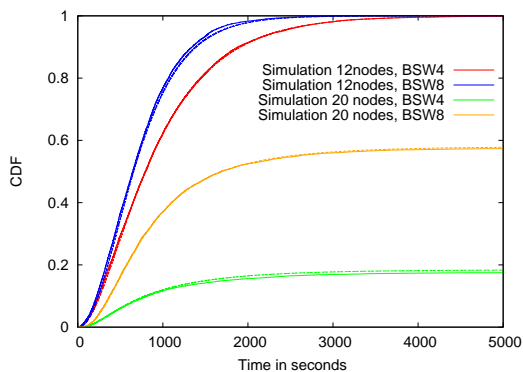} 
		\caption{Results obtained for cases \#2 and \#3.}
		\label{results_20n} 
		\end{center}
\end{figure}